\documentclass[twoside,11pt]{article}
\usepackage{irrj2e}

\usepackage{amsmath,subfigure,xurl,xspace}

\let\svthefootnote\thefootnote
\newcommand\freefootnote[1]{%
  \let\thefootnote\relax%
  \footnotetext{#1}%
  \let\thefootnote\svthefootnote%
}

\usepackage{lastpage}
\irrjheading{1}{2025}{29-46}{9/24}{03/25}{19625}{Ian Soboroff}
\ShortHeadings{Don't Use LLMs to Make Relevance Judgments}{Ian Soboroff}
\firstpageno{29}

\title{Don't Use LLMs to Make Relevance Judgments}
\author{Ian Soboroff\\
  National Institute of Standards and Technology\\
  Gaithersburg, Maryland, USA\\
\url{ian.soboroff@nist.gov}}

\editor{Vanessa Murdock, Djoerd Hiemstra}

\begin{document}
\maketitle

\begin{abstract}
  Relevance judgments and other truth data for information retrieval (IR) evaluations are created manually.  There is a strong temptation to use large language models (LLMs) as proxies for human judges.  However, letting the LLM write your truth data handicaps the evaluation by setting that LLM as a ceiling on performance.  There are ways to use LLMs in the relevance assessment process, but just generating relevance judgments with a prompt isn't one of them.\footnote{This article reflects the views of the author and not necessarily those of NIST or of the U. S. Government.}
\end{abstract}

\section{Introduction}

The Text Retrieval Conference (TREC) is a community evaluation and dataset construction activity sponsored by the U.S. National Institute of Standards and Technology (NIST).  TREC has run annually since 1991.  TREC is divided into {\it tracks} which embody specific search tasks.  The canonical TREC task is {\it adhoc search}, with searches against a static set of documents, each search returning a single ranked list of documents.  An individual search instance is called a {\it topic} and expresses the user's information need in long form rather than providing a query.  The relevance judgments, or {\it qrels}, maps each topic to the documents that should be retrieved for it.  The combination of the document collection, the topics, and the relevance judgments is called a {\it test collection}. 

The relevance judgments representing ground truth are created collaboratively between participants and NIST using a process called {\it pooling}~\citep{pooling}.  TREC participants use their IR systems to return the top $K$ documents for each topic.  The union of the top-ranked $k \ll K$ documents from each participant system is called the {\it pool}.  The documents in the pool are reviewed by the person who invented the topic, and they decide which documents are relevant and which are not.  Using the qrels as labels we can compute various measures of retrieval effectiveness such as precision and recall.  The test collections let researchers rapidly innovate new search algorithms in a laboratory setting before deploying them to a live system.  More information about TREC can be found in~\cite{trec-book}, a book covering the first ten years of the program.

The process used in TREC descends from the Cranfield indexing experiments conducted by Cyril Cleverdon in the 1960s, and so we say TREC is following the Cranfield paradigm~\citep{cranfield}.  Central to the Cranfield paradigm is a set of assumptions that simplify the search problem: the document collection and information needs are fixed, all documents are labeled as relevant or not relevant to every query, relevance is modeled by topical similarity, the relevance of a document is independent of the relevance of any other document, there is a single query that is answered with a single ranked list, and the relevance judgments are representative of the user population~\citep{voorhees-phil}.  TREC can be thought of as a community effort in pushing the bounds of the Cranfield paradigm.  Complete judgments were replaced with the pooling procedure, which has been shown to be sufficient for measuring the pooled systems and also useful for measuring systems which were not pooled for evaluation, as long as certain properties are maintained~\citep{DBLP:conf/trec/Harman95, Zobel1998, Buckley2007}.  Likewise, many TREC tracks push back on the notion of static documents and information needs~\citep{DBLP:conf/trec/FrankKRVS14,DBLP:conf/trec/CarteretteKHC14}, relevance as topical similarity~\citep{DBLP:conf/trec/CraswellH04,DBLP:conf/trec/BalogSV11}, single rankings~\citep{DBLP:conf/trec/Owoicho0AATV22,DBLP:conf/trec/AliannejadiACDA23}, and independent relevance~\citep{soboroff-harman-2005-novelty}.

Making the relevance judgments for a TREC-style test collection can be complex and expensive.  Relevance assessing at NIST for a typical TREC track usually involves a team of six contractors working for 2-4 weeks.  Those contractors need to be trained and monitored.  Software has to be written to support recording relevance judgments correctly and efficiently.  Experience in both the technical and human aspects of the process counts for a lot, which is why we run evaluation campaigns rather than everyone building their own test collections.  Evaluation campaigns are infrastructure for IR research.

The recent advent of large language models that produce astoundingly human-like flowing text output in response to a natural language prompt has inspired IR researchers to wonder how those models might be used in the relevance judgment collection process~\citep{dagstuhl-report,llm4eval-cacm}.

At the ACM SIGIR 2024 conference, a workshop ``LLM4Eval'' provided a venue for this work, and featured a data challenge activity where participants reproduced TREC deep learning track judgments, as was done by~\cite{thomas2024}. I was asked to give a keynote at the workshop, and this paper presents that keynote in article form.

The bottom-line-up-front message is, don't use LLMs to create relevance judgments for TREC-style evaluations.

\section{Automatic evaluation}

The idea of automatic evaluation for information retrieval came from a paper I wrote with Charles Nicholas and Patrick Cahan in 2001~\citep{Soboroff2001}.  I had been reading Ellen Voorhees well-known SIGIR paper from 1998~\citep{Voorhees1998} which shows experimentally that while people differ in their judgments of relevance, those differences don't affect the relative ordering of systems in a TREC evaluation.  Surprised by this result, I wondered what would happen if the relevance judgments were randomly sampled from the pools.  Certainly, TREC relevance judgments aren't random, but how much can they vary towards random and still rank systems equivalent to the official system ranking?

A representative example result from that work is shown in Figure~\ref{t8-trecstyle}.  The single $+$'s are official scores from TREC, and the $\times$'s with whiskers up and down are scores obtained with random documents from the pool labeled as relevant.  The points are ordered according to their official TREC scores.  The key point to notice is that, using random judgments, the best systems (those with the highest MAP, on the left) look like the worst systems (on the right).

\begin{figure}
  \includegraphics[width=\textwidth]{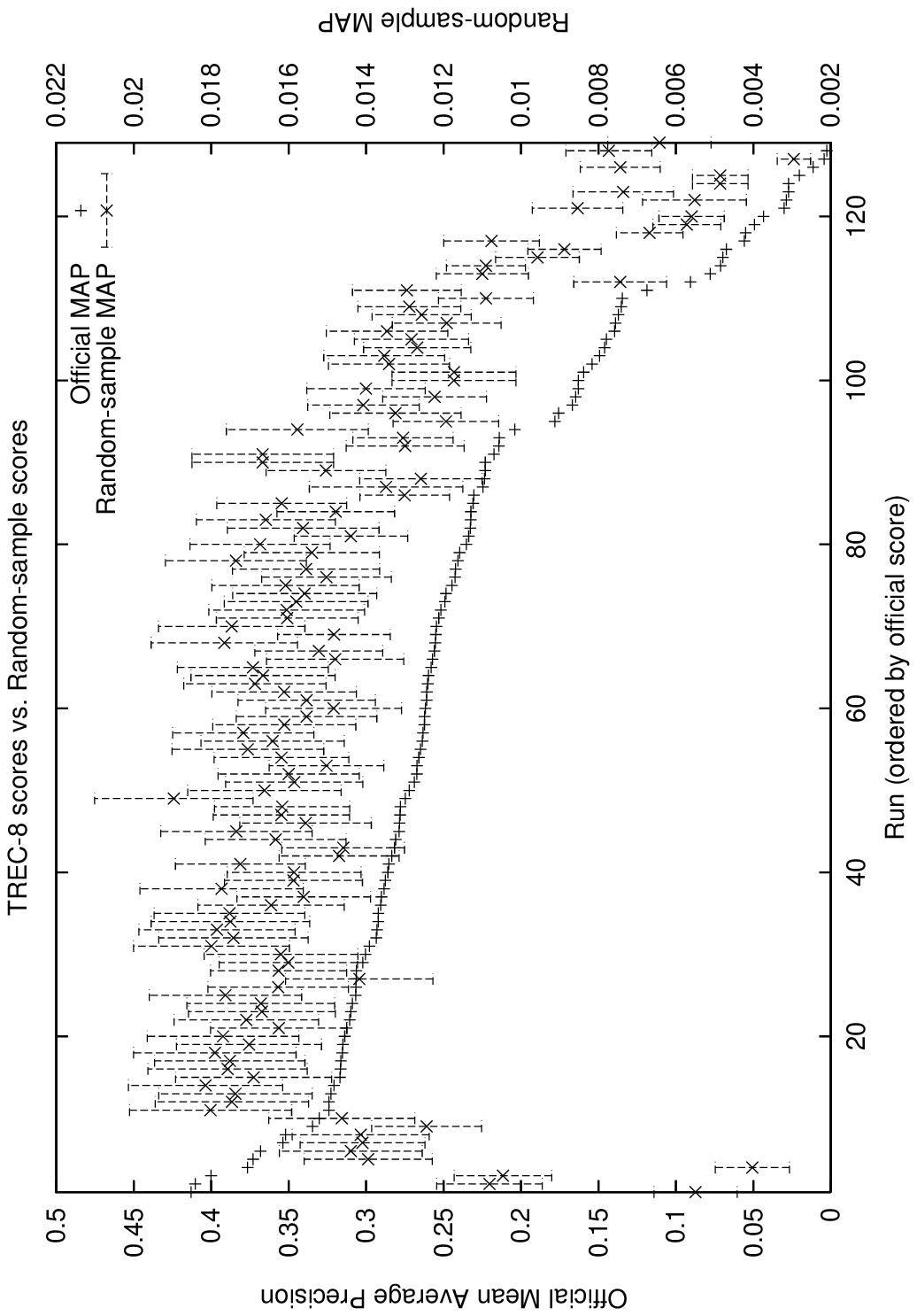}
  \caption{\label{t8-trecstyle}Sample result from \cite{Soboroff2001}, TREC-8, TREC-style pooling to depth 100.}
\end{figure}

Automatic evaluation in this sense means making relevance judgments using an algorithm rather than people, as opposed to inducing relevance from implicit behavior cues or history.  In published papers proposing automatic evaluation methods for IR, the quality of these methods is quantified by comparing the ordering of systems induced by the automatic method to the official TREC ranking based on manual assessments.  That is, the qrels are used to score each system using some metric, and the systems are then ordered by their score.  The common metric used is Kendall's tau ($\tau$), a correlation measure between rankings.  Some researchers also use the more familiar Spearman's rho ($\rho$), a correlation measure that takes into account the distance between the points and not just their rank order.  Others have proposed versions of these correlations that emphasize the upper part of the ranking (i.e. the best systems)~\citep{Yilmaz2008}.  I have heard a discussion of a variation of tau where a swap in position between two systems only happens if they are significantly different according to some statistical test, which would take advantage of the fact that the points in the ranking represent average performance over a set of topics.

\cite{Aslam2003} published a short paper that explained my 2001 results.  No matter how many systems retrieve a given document, it is only added to the pool once.  There are many more nonrelevant than relevant documents, so a given relevant document was likely to have been retrieved by more than one system.  (The mean number of systems retrieving a nonrelevant document in TREC-8 is 3, versus 11 for a relevant document.)
By selecting the relevant documents at random, I was implicitly selecting documents retrieved by many systems.  So my 2001 paper shows the results of a popularity contest.
Under this approach, the worst systems and the best systems both look bad, because they fail (or succeed) by retrieving documents that other systems don't find.  Another way to think about Aslam and Savell is that by using the output of a system as the ground truth, I am measuring the similarity of the two systems, how close the retrieval system is to the model that created the ground truth.

Around that time, the BLEU automatic metric for machine translation~\citep{papineni-etal-2002-bleu} and the ROUGE automatic metric for summarization~\citep{lin-2004-rouge} were published.  These measures compare system-generated outputs, such as translations and summaries, by the overlap of word $n$-grams with a model or reference output.  ROUGE worked well for extractive summarization, where a summary is produced by cutting and pasting sentences from the source documents, but less well in a generative setting where word choice could vary quite a bit from the original documents.

\section{Machine learning and predicting from examples}

A more successful method of imputing relevance comes in the form of relevance feedback and more sophisticated machine learning algorithms.  In these cases, examples of relevant (and possibly nonrelevant) documents are used to train a model to predict the relevance of other documents.

Relevance feedback (RF) in the vector space model, developed by Joseph Rocchio and Eleanor Ide in the mid to late 1960s as part of the SMART system~\citep{salton-mcgill-1983}, may be the first on-line machine learning algorithm.\footnote{Rocchio did not describe his method as machine learning, but he did develop a theory that his relevance feedback method builds an optimal query.}  In RF, the user executes an initial search and identifies one or more documents in the search results as being relevant or irrelevant.  The terms in the query are augmented and reweighted based on the feedback, and the refined query is executed to rank the remaining documents in the collection.  Since then, it has been adopted as a general IR technique rather than a specific algorithm and has been instantiated within nearly every ranking model.  Currently, the most common implementation of RF uses the BM25 ranking algorithm with the RM3 method of term weighting~\citep{DBLP:conf/trec/JaleelACDLLSW04}.  Pseudo-relevance feedback (PRF) is a modification in which instead of indicating relevance by the user, some number of documents ranked highest in the initial ranking are assumed to be relevant for a round of relevance feedback~\citep{DBLP:conf/trec/BuckleySAS94}.  Relevance feedback is one of the most successful techniques in information retrieval, producing large improvements in performance.  PRF is somewhat similar except that for some topics it fails because the initial retrieval is off base for some reason.

\cite{Buttcher2007} proposed using the TREC relevance judgments to predict the relevance of unjudged documents retrieved by unpooled systems, and also as a method for expanding the set of relevance judgments overall.  They use the qrels to train a binary classifier and then apply that to documents that were not judged but were ranked above the pool depth of TREC's pools.  Anecdotally, this technique did not perform as well when all retrieved documents (down to rank 1000) were predicted, so there is something to restricting predictions to those documents that are already ranked highly by the search ranker.

\cite{Rajput2012} describe an iterative method using {\it nuggets}.  A nugget here is a manually selected passage from a relevant document.  Starting with manual nuggets, the process identifies new high-probability shingles as new nuggets and uses those to predict the relevance of other documents.  Their use of the term ``nugget'' is different than how the term is used for evaluation of summarization and question answering; summarization nuggets are atomic pieces of information which must be manually aligned to the generated summary, whereas these nuggets, being strings or shingles, are automatically matched.  Nuggets in~\citep{Rajput2012} are essentially lexical patterns that identify relevant documents.

The BERTScore metric \citep{Zhang*2020BERTScore:} computes token similarity between a generated and reference output using BERT embeddings.  We can think of this as the LLM equivalent of BLEU, using embeddings instead of n-grams.

\section{LLM-based predictions of relevance}

Modern large language models (LLMs) have inspired a new approach, where a topic and document are embedded in a prompt, which is then fed to an LLM that outputs some indicator of relevance.  LLMs may be fine-tuned with relevance examples, or other relevant documents may be included in the prompt, but otherwise no examples are used as in the supervised learning methods above.

\cite{thomas2024} describe using LLMs to predict relevance in TREC collections as well as for search results from a major commercial search engine.  They develop prompts at a number of levels of richness.  In their web search results, they find that the generated judgments ``have proved more accurate than any third-party labeler, including staff; they are much faster end-to-end than any human judge, including crowd workers; they scale to much better throughput; and of course are many times cheaper.''  The paper describes results on TREC data in greater detail, and there is an extended discussion of their prompts and their evolution.

\cite{MacAvaney2023} used nearest-neighbors, classifiers, and LLM prompts to elicit relevance judgments to supplement judgments in MS-MARCO~\citep{MSMARCO-2016}, a collection where there is only a single manually-adjudicated relevance judgment per query.  By basing a system's performance measurement on more than one document, IR metrics are found to be more stable.

\cite{Alaofi2024} investigated the agreement between LLMs and TREC assessors and found that  LLM false positive decisions seemed to be related to the presence of query terms in the passage being assessed.  A false positive is where the LLM votes that the passage is relevant, but the human assessor judged it to not be relevant.  In many of these cases, the false positive passages included terms from the query, despite not being relevant.  This seems to imply that despite the richness of the language model, lexical cues can influence the decision more than the true meaning of the text.

Outside of the information retrieval domain, researchers seem to be eagerly jumping on a bandwagon for LLM-based automatic evaluation.  As one example, \cite{lin-chen-2023-llm} employ prompts to gauge generated responses in open-domain dialogues, and compare results to other automatic evaluation techniques, some of which use the LLM to identify properties of good responses \citep{mehri-eskenazi-2020-usr,mehri-eskenazi-2020-unsupervised} and others which use the LLM to directly assess dialog responses \citep{chen-etal-2023-exploring-use,fu-etal-2024-gptscore}.  None of the comparison metrics is validated against manual labels of the dialogues in question.

The search for automatic metrics is long and has made use of new algorithms as they have been developed.  There is a real need for automatic metrics, because manual assessment is slow and hard to scale.  When the labels are created zero-shot, specifically meaning that the evaluation model is operating at the same degree of data exposure as the systems being measured, the evaluation reduces to comparing the performance of the system to the model, not to human performance.  When the evaluation model has more knowledge than the systems being measured, for example relevance judgments on the topics in the test set, then the model may produce an evaluation that can stand as a useful measurement, a comparison to something more than just another system.  When the evaluation model is making use of outside knowledge, for example in \cite{mehri-eskenazi-2020-unsupervised}, then the situation depends on the systems being measured.  The following sections elaborate this argument.

\section{Retrieval and evaluation are the same problem}

Asking a computer to decide if a document is relevant is no different than using a computer to retrieve documents and rank them in order of predicted degree of relevance.  In both cases, the algorithm makes the assessment of relevance.

A retrieval system, or a relevance model, is a model of relevance given available data.  The system is trying to predict which documents are relevant and which documents are not.  Even though real systems might try to optimize a pairwise or listwise output or compute a degree of relevance of a document or a search engine result page, it is useful to think of all these processes as predicting relevance.

During relevance assessment, we are asking the assessor to decide whether documents are relevant or not.  This, too, is essentially a prediction of relevance.  It's a well-informed prediction since the person is reading the document and often composed the information need, but since the task is artificial, the assessor is basically saying that they would include this document in a report on the topic, a report which they don't ever actually write.  We can call that a prediction too.

We use one set of predictions, the relevance judgments, to measure the performance of the other set of predictions, the system outputs.  In doing so, we declare the relevance judgments to be truth.  In fact, you can switch the two sets of predictions, declare the system output to be truth, and measure the ``effectiveness'' of the assessor compared to that of the system.  All evaluations which compare a system output to an answer key are making a measurement with respect to the answer key, not with respect to the universe.

Since both retrieval systems and relevance assessors are making predictions of relevance, evaluation and retrieval are the same problem.  We can imagine a very slow system that would have a human read every document and assess its relevance given a query.

John Searle's ``Chinese Room'' thought experiment\footnote{\url{https://en.wikipedia.org/wiki/Chinese_room}} posits a person in a box who receives questions through a slot and delivers answers out the slot.  The questions and answers are in Chinese, a language which the person does not read or speak.  Rather, the person follows a sophisticated set of instructions for generating an answer from a question, in Chinese, by manipulating symbols on the paper.  Thus, the box appears to understand and communicate in human language but is basically a computer.  A mechanical Chinese Room can be implemented with an LLM chatbot.  Construct a prompt of a topic statement and a document and ask for the LLM to say relevant or not, for every document in the collection.  Asking the language model about relevance is the mirror of evaluation.

If we believed that a model was a good assessor of relevance, then we would just use it as the system.  Why would we do otherwise?  We don't use human assessors that way, because it doesn't scale.  LLMs in 2024 don't scale, but that feels like an engineering problem more than something fundamental; we will probably solve this with better hardware and smaller models.

Since both retrieval and evaluation are prediction activities it seems natural to apply machine learning to both.  The predictions don’t happen in isolation: systems know about collection frequencies and click patterns that inform the ranked list, and assessors have experience and world knowledge that informs their labels.  Machine learning, the field where we train prediction systems by example, clearly has a role to play here.

As with any prediction, there are errors of omission and commission (or false negatives and false positives if you prefer), and those errors represent a maximum discriminative ability of those relevance judgments to distinguish systems.  I will dive into this in more detail in the next section.

\section{The ceiling on performance}

Whatever we use as the answer key represents both an ideal solution and a ceiling on measurable performance.  No system can outperform the evaluation's answer key.  When the answer key is made up of human-assigned labels, then we are saying that human performance is the ideal we are aiming for, and we can’t measure something better than that performance. Likewise, when the answer key is created by a machine learning model or some other mechanical process, we are saying that the model represents the idea we are aiming for, and we can’t measure something better than the performance of that model.  This is the critical flaw with LLM-sourced relevance judgments.

An IR test collection is a 3-tuple:
\[ C = \{ D, S, R \} \]
where $D$ is a set of documents, $S$ is a set of search needs, and $R$ is a function $R: S \mapsto D$ that maps search needs to relevant documents.
In the original Cranfield collections, there is a value of $R$ for every document $d$ and search need $s$.  In the TREC collections, most of those pairs are unknown and pooling lets us assume that an unjudged document is likely not relevant.

A retrieval model produces a ranking of documents $d_n: d_n \in D$ in order of predicted relevance to the search need $s$:
\[ A(s, D) = \{ d_n~\forall~d_n \in D \} \]
where the set here is an ordered set, a sequence of documents where each document appears once.  The entire document collection is ranked although in practice we cut off the ranking much earlier.

An evaluation function $E(A,R)$ computes a real number from a $k$-prefix of the ranked list $A^k$.  Often in TREC $k = 1000$ but some measures set $k$ much lower to focus on the top of the list.  If the number of relevant documents for $s R_s >= k$, then a system can produce a ranking whose prefix consists only of relevant documents.  Thus in practice we try to have search needs with many fewer relevant documents, and a fundamental difficulty with enormous collections like ClueWeb is that we can easily find thousands of relevant documents and still worry that we have not found them all~\citep{Buckley2007}.

The relevance judgments in a Cranfield experiment are a model of human behavior, and since we are trying to build systems that understand information needs and documents as well as humans do, they model ideal retrieval performance.  The evaluation function $E$ is typically defined to be maximized by an ideal ranking, for example if all relevant documents are ranked ahead of any irrelevant documents.
If you take the relevance judgments and turn it into a run by first listing all the relevant documents and then padding the listing to $k$ with irrelevant documents, it gets perfect scores on the appropriate metrics.

Historically, this was the goal of IR performance.  IR systems are meant to augment people by scaling up their ability to understand information, and so the performance of people is the ideal.

This ideal is also a limit on what Cranfield can measure.  Under $R$, the best possible ranking
\[ A(s, D) \mapsto \{ +, +, +, +, ... -, -, - \} \]
orders the relevant documents ahead of any irrelevant documents.  The order of relevant documents among themselves, and irrelevant documents among themselves, are not important: there is a very large number of equivalently ideal rankings by permutations among the relevant and irrelevant documents.  For graded relevance regimes, this ranking orders documents by their rated degree of relevance, where those degrees are positive integers, zero for not relevant, and perhaps negative numbers for other poor outcomes like spam, and within each relevance degree or category the documents can be permuted to create equally ideal ranked lists.  If two or more categories are equivalently valuable, we can replace them with a superset including all equivalently valuable documents.  Without loss of generality, moving forward I will assume that rankings can have all the relevant documents ahead of all the irrelevant documents for any relevance construct.

\begin{theorem}[ideal rankings]
  Let $C$ be a test collection $(D, S, R)$ where $R: s \mapsto d$ maps search needs to relevant documents $\{+,+,+, ...\}$. Let $A(s, D)$ be a ranking function that produces a ranking of documents $\{d_n \in D\}$ for a search need $s$.  Let $E: A(s, D), R \mapsto \mathbb{R}$ be an evaluation metric that computes a real number representing the quality of the ranking $A$ given the relevance judgments $R$. Then, we can define the \textbf{ideal ranking} as
 \[ A(s, D) \mapsto \{ +, +, +, +, ... -, -, - \} \]
the ranking the places the relevant documents ahead of any irrelevant documents.  The ideal ranking maximizes $E$, and there does not exist any ranking $A'$ that obtains a higher value than $A$ of $E$ subject to $R$.
\end{theorem}
\begin{proof}
  Suppose a ranking $A'$ that has one or more relevant documents that are not in ideal ranking $A$.  For $A'$ to be ideal, these extra relevant documents must appear at the head of the ranking.  However, the ideal is defined subject to $R$, the full set of relevance judgments, and $A$ is already defined to be ideal.  If there are extra relevant documents missing from $A$, then $A$ is not ideal.  If the new ``relevant'' documents are not in $R$, then $A'$ can't be ideal either.  So $A$ and $A'$ must have the same set of relevant documents in their ranking, and $E(A,R)$ and $E(A',R)$ are equal and maximize $E$.
\end{proof}

If we imagine we have a system that is better than a human, for example by finding relevant documents that are not in the relevance judgments or correctly ranking a document which was assessed incorrectly, that system will score less than perfectly when we score it using the human relevance judgments.  This must be the case, because unjudged documents are assumed to be not relevant, and the documents found by this novel system are either absent from the relevance judgments or judged non-relevant when they should have been marked relevant.  The system has retrieved documents which \textbf{ according to the evaluation relevance judgments} are not relevant.  And so this top-performing system is measured as performing less well than it does.

This is reflected in my 2001 paper, and with other papers that came later, but exemplified by Figure~\ref{t8-trecstyle} above.  The (true) top systems are under-ranked by the ``model'' of relevance.  This must be true for any model of relevance that generates relevance judgments, be it human or machine.  We cannot measure a system that is better than the relevance judgments.  Or, rather, the evaluation can't distinguish such a system from one that performs less than perfectly.

As a counterexample, ~\cite{Buttcher2007} trained a model using an incomplete set of manual judgments to classify a larger set of documents automatically, improving the collection. In this case the evaluation model is privileged in comparison to the runs, in that it has relevance information that they do not.  Relevance feedback nearly always improves performance, so we would expect a hybrid set of judgments like these to have the possibility of outperforming evaluation using the shallow judgments.  This is outperforming the original human but only doing so by retrospective use of human relevance data.\footnote{We still have a grounding problem, in that you may not believe that the model makes accurate predictions.  The process can be improved by doing a second stage of relevance assessments on the classifier outputs in order to estimate the classifier error rates.}

And so when the relevance judgments are created by a person, the model can't exceed the human ideal.  If we had a model that had ``super-human'' performance, we would just make our IR system use that model.  In the current state of the art, the most advanced LLMs are used as components of systems that may be hypothetically measured by relevance judgments generated from the same models.  Those systems cannot perform better than the model generating the relevance judgments.

Obviously, the human that created the relevance judgments is not entirely ideal.  The assessor is not all-knowing, all-seeing, all-reading with perfect clarity.  Assessors make mistakes, and TREC participants are fond of finding them.  More importantly, the assessor is only one person; someone else with the same information need would make different judgments.  If we compared the assessor's judgments to those of a secondary assessor by pretending that secondary assessor is a run, it would necessarily perform less than perfectly.

That means that even if we imagine that systems exist which perform better at the task than humans do, we can't see that improved performance in a Cranfield-style evaluation.  
This follows from the ideal ranking theorem above.  It must also be true for \textbf{any evaluation} where we are comparing a system output to a ``gold standard,'' for example in machine learning or natural language processing, because the gold standard represents ideal performance, and by the ideal ranking theorem, no ranking can be measured as better than the truth data.

The so-called ``super-human'' performance observed on benchmark datasets\footnote{For example, \url{https://openai.com/index/planning-for-agi-and-beyond/} and \url{https://venturebeat.com/ai/google-deepmind-unveils-superhuman-ai-system-that-excels-in-fact-checking-saving-costs-and-improving-accuracy/}  For a contrasting view, see~\cite{Tedeschi2023}.} is actually just measurement error.  Super-human performance would be scored as less than ideal by the established ground truth, because performing better than a human entails making different decisions than those in the ground truth.

Some benchmarks are capable of showing super-human performance by differentiating between the humans performing the task and the humans that create the answer key.  For example, an LLM may perform better than many people on a standardized test, but we can measure that because the humans taking the test are not the source of the correct answers.  Likewise tests of solving analogies or complex math problems.  In IR evaluations, we are only comparing to the answer key, not another person’s attempt to recreate the answer key.\footnote{We wouldn’t do that because assessor disagreement is reality. There are no absolutely correct answers outside the Cranfield room.}

To summarize, you should not create relevance judgments using a large language model, because:
\begin{itemize}
\item You are declaring the model to represent ideal performance, and so you can't measure anything that might perform better than that model.
\item The model used to create relevance judgments is certainly also used as part of the systems being measured.  Those systems will evaluate as performing poorly even if they actually improve on the model, because improving on the model means retrieving new relevant unjudged documents that aren't in the answer key.
\item When the next shiny model comes out, it will measure as performing less well than the old model, because it necessarily must retrieve unjudged documents or ones judged incorrectly as not relevant.  And so the relevance judgments can only measure systems that perform worse than the state of the art at the time the relevance judgments were created.\footnote{This might be a good thing.  Since new models will look worse than old models, when their developers run them through leaderboard-style benchmarks they will cast them aside because they seem to perform poorly.  Then we won't ever have to worry about having a new model.}
\end{itemize}

\section{Limitations}

The argument in this paper makes the case that using an LLM to generate the ground truth for an IR evaluation results in a substandard evaluation.  However, it could certainly be the case that LLMs could play different roles in the evaluation process than inventing the answer key.

LLMs can be used to create the answer key if they have more knowledge about relevance than the systems being measured.  If the evaluation model has access to privileged information, for example by being fine-tuned on manual relevance judgments on the evaluation topics, then those relevance judgments should still be able to measure systems that use the untuned model.  While it might be tempting to assume that the LLM has information about relevance in the training data, we should avoid this assumption since we don’t have access to the training data.

Blessing the evaluation model with extra relevance information is what makes B\"uttcher et al's~\citep{Buttcher2007} results work: the model is trained on relevance data, and so the trained model has the advantage over any system it is measuring that doesn’t have access to that relevance data.

We actually already knew this: the fact that relevance feedback improves retrieval is a basic result in IR.  If we have relevance information gleaned from many systems as we do when pooling, the outputs will perform better than any individual system and thus we can measure any individual systems with less information about relevance than the collective pool.

This still has the problem that the evaluation isn't future-proof: we might have a new model that outperforms the relevance feedback of the prior generation.  We haven't seen this yet when the collections are pooled from older systems only, if the collection is well-judged~\citep{voorhees-soboroff-lin-2024}.  Or new models might have TREC triples (topic, document, relevance) as an explicit component in their training data and be able to make use of that in a retrieval setting.

One simple idea that seems promising is to employ a LLM to follow the assessor and look out for mistakes.  This can't be done by simply asking, ``Did you mean `relevant' instead?'' since people are primed to trust the computer more than they should~\citep{LOGG201990, bogert21}.\footnote{\cite{LOGG201990} is also interesting as much behavioral literature four to five years prior seemed to find that people distrusted algorithmic recommendations.  Perhaps our perspectives are changing with exposure.  But see also~\cite{doi:10.1126/science.1207745} on search's effects on memory.}  But it may be possible to automate a quality-control process using a model.

There are evaluation activities that don’t involve creating an answer key.  For example, in a user study, researchers observe user behavior and analyze those observations to draw conclusions about the experiment.  LLMs might be useful in supporting the observational process (perhaps by transcribing mouse movements and clicks in a readable way) or the analysis process (much as we use statistical models to determine significance).

At the SIGIR workshop, a questioner asked if an LLM-generated evaluation might still be useful even given its flaws.  For example,~\cite{thomas2024} found LLM judgments to be as useful as crowdsourced judgments, but not better than curated judgments from a trained team.  In their setup, crowd judgments represented a low rung in a tiered hierarchy of relevance judgments and system measurements.  If the judgments are not meant to support a rigorous evaluation but rather as noisy training data, then the LLM judgments may be useful.  But if the LLM creating the truth data is part of the search system, or is of an older generation than the search system, the results may under-report performance and not be able to distinguish improvements, as shown above.  In all cases it should be kept in mind that the ideal used as a comparison point is not human performance, but model performance.

\section{Conclusions}

I have discussed the limitations of using models to create relevance judgments.  You don't want to do that, because then you have limited what you can measure to the level of the generating model.  If the generating model is also part of the evaluated systems, you are stuck in a loop, or perhaps falling into a bottomless pit.

This is similar to model collapse~\citep{guo2024}.  When you train the model using its own outputs, the performance of the model decreases.  The collapse mechanism is the measurement error that comes from generating the truth data using the evaluated system.  In this case, evaluation is a loss function computed based on the generated truth data.

This doesn't mean that LLMs can't allow us to do amazing things.  As someone who got his start in IR working with LSI~\citep{Deerwester1990}, which is essentially an optimal linear embedding, I am very excited by the idea of nonlinear embeddings.  IR systems that use LLMs to surmount the vocabulary boundary have enormous promise for real users.

All models have limits, and humans do too.  If we want to use the model to evaluate performance, we first need to consider if we are doing something past the ability of the model as used in that evaluation paradigm.  The relevance judgments barrier is a fundamental limitation of evaluations that measure systems against ground truth.

\section{Acknowledgments and Disclosure of Funding}

I thank Ellen Voorhees and Rikiya Takehi for their comments on the talk and on early versions of this paper.  I also thank the attendees of the SIGIR 2024 LLM4Eval workshop for their insightful questions and continuing discussions.

No external funding was received in support of this work.

The TREC activity has been annually reviewed by NIST’s Research Protection Office and determined to not be human subjects research.

Any company, product or service mentioned in this paper should not be taken as an endorsement of that company, product, or service by NIST.  Nothing in this paper should be read as a comparison to or among commercial products.

\bibliography{refs.bib}

\end{document}